\documentclass[aip,jmp,nofootinbib]{revtex4-1}
\usepackage{graphics}
\usepackage{graphicx}% Include figure files
\usepackage{dcolumn}% Align table columns on decimal point
\usepackage{bm}% bold math
\usepackage{mathrsfs}
\usepackage{pstricks}
\usepackage{color}
\usepackage{slashed}
\usepackage{amsmath}
\usepackage{epsfig}
\usepackage{amsfonts}
\usepackage{amssymb}
\usepackage{amsthm}
\newtheorem*{theorem}{Theorem}
\def\beq{\begin{equation}}
\def\eeq{\end{equation}}
\def\bea{\begin{eqnarray}}
\def\eea{\end{eqnarray}}

\def\diag{\textrm{diag}}
\def\Re{\textrm{Re}}

\def\E{\textrm{E}}

\begin{document}

\title{Riemann zeta zeros and prime number spectra in quantum field theory}
\author{G. Menezes}
\email{gsm@ift.unesp.br}
\affiliation{Instituto de F\'{i}sica Te\'orica, Universidade Estadual Paulista, S\~ao Paulo, SP 01140-070, Brazil}

\author{B. F. Svaiter}
\email{benar@impa.br}
\affiliation{Instituto de Matem\'atica Pura e Aplicada, Rio de Janeiro, RJ 22460-320,
Brazil}
\author{N. F. Svaiter}
\email{nfuxsvai@cbpf.br}
\affiliation{Centro Brasileiro de Pesquisas F\'{\i}sicas, Rio de Janeiro, RJ 22290-180,
Brazil}

\begin{abstract}

The Riemann hypothesis states that all nontrivial zeros of the zeta function lie in the critical line $\Re(s)=1/2$. Hilbert and P\'olya suggested that one possible way to prove the Riemann hypothesis is to interpret the nontrivial zeros in the light of spectral theory. Following this approach, we discuss a necessary condition that such a sequence of numbers should obey in order to be associated with the spectrum of a linear differential operator of a system with countably infinite number of degrees of freedom described by quantum field theory. 
The sequence of nontrivial zeros is zeta regularizable. Then, functional integrals associated with hypothetical systems described by 
self-adjoint operators whose spectra is given by this sequence can be constructed. However, if one considers the same situation with primes numbers, the associated functional integral cannot be constructed, due to the fact that the sequence of prime numbers is not zeta regularizable. Finally, we extend this result to sequences whose asymptotic distributions are not ``far away" from the asymptotic distribution of prime numbers.

\end{abstract}

%\date{\today}


\maketitle

\section{Introduction}

The Riemann zeta function $\zeta(s)$ defined by analytic continuation of a Dirichlet series has a simple pole
with residue $1$ at $s=1$, trivial zeros at $s=-2n$, $n=1,2,...$ and infinitely many complex zeros $\rho=\beta+i\gamma$ for
$\beta,\gamma \in  \mathbb{R}$ and $0<\beta<1$.
The Riemann hypothesis is the conjecture that $\beta=1/2$~\cite{riem}. Another conjecture is that all zeros are simple.
Both conjectures are unsolved problems in mathematics and much effort has been expended using various
approaches to prove or disprove such conjectures. Hilbert and P\'olya suggested that one way to prove the Riemann
hypothesis is to give a spectral interpretation to nontrivial zeros of $\zeta(s)$. The nontrivial zeros could be the
eigenvalues of a linear self-adjoint operator in an appropriate Hilbert space. For a nice introduction to the
Riemann hypothesis, see the Refs.~\cite{borwein, rev}.

In recent years the behavior of solutions of the Hamilton
equations for different systems have been analyzed using methods
of nonlinear mechanics. Research in this direction started
with Gutzwiller \cite{gutzwiller1,gutzwiller2} who
presented a formula which enables one to calculate the spectral
density of chaotic systems, a trace formula that express the
density of states as a sum over the classical periodic orbits of the system.
Subsequent works of Bohigas, Giannoni, Schmit \cite{bo,bohigas} and also
Berry \cite{berry} raise the conjecture that a quantum energy spectrum of
classically chaotic systems may show universal spectral correlations
which are described by random matrix theory~\cite{mehta,for,for2}. Another important conjectured was stated by Montgomery~\cite{mo} and supported
by numerical evidences~\cite{mo2}: the distribution of
spacing between nontrivial zeros of the $\zeta(s)$ function is statistically identical to the distribution of
eigenvalue spacings in a Gaussian unitary ensemble. Such a conjecture prompted many authors to consider the
Riemann hypothesis in the light of random matrix theory and quantum mechanics of classically chaotic systems ~\cite{bu,buu,siam,bu2,bourgade}.

Since nonseparability in wave mechanics leads to chaotic systems in
the limit of short wavelengths, from the above discussion one cannot disregard the possibility that a quantum field model
described by a nonseparable wave equation for a given boundary condition is able to reproduce the statistical properties of the nontrivial zeta zeros. Other possibility that could shed some light in the spectral interpretation for the zeros is the study of field theory in disordered media, such as wave equations in random fluids or amorphous solids~\cite{ishimaru,ping,krein1,new} or also nonlinear dielectrics~\cite{annals}. For instance, it was shown that the level spacing distribution of disordered fermionic systems can be described using random matrix theory~\cite{efetov1,efetov2}. Therefore, one could naturally argue that further progress on the Hilbert-P\'olya approach can be achieved investigating the kind of systems discussed above. If one is willing to employ quantum field theory functional methods in order to describe a system with countably infinite number of degrees of freedom, there is a necessary condition that should be met: all path integrals must be finite. Here we use the zeta regularization procedure, a standard technique for regularizing quadratic path integrals.  We remark that systems with the spectra of the nontrivial zeros of the $\zeta(s)$ function have been discussed in the literature before. The statistical properties of a Fermi gas whose single-particle energy levels are given by these zeros  were investigated in~\cite{le,le2}. Also, in order to solve questions related to the number theory using statistical-mechanics methods, some authors introduced number theory and prime numbers in quantum field
 theory~\cite{stn,sss,bakas,julia,sp,spector}.

The aims of the present paper are the following. First we investigate if a quantum field theory model can describe some hypothetical physical system with infinite degrees of freedom, whose spectrum is given by the sequence of prime numbers.
We regularize the product of a sequence of numbers by means of the zeta regularized product associated with them. If these numbers can be associated with the spectrum of a linear differential operator, the zeta regularized product is the determinant of the linear operator~\cite{se,ray,dowker,ha,vo,vor,quine,dune}. Using the fact that the sequence of prime numbers is not zeta
regularizable~\cite{ga}, one cannot conceive a physical system with a pure prime number spectrum that could be described by quantum field theory. This result has been recently generalized by Andrade~\cite{andrade} for other number-theoretical sequence of numbers. Afterwards, prompted by this result we prove that sequences whose asymptotic distribution is close to that of the prime numbers are also not zeta regularized. Therefore, these sequences cannot be associated with the spectrum of a self-adjoint operator in quantum field theory. Then we show that the sequence of the nontrivial zeros of the Riemann zeta function can in principle be interpreted as being the spectrum of a self-adjoint operator of some hypothetical system with countably infinite number of degrees of freedom described by quantum field theory. Finally, we discuss possible relations between the asymptotic behavior of a sequence and the analytic domain of the associated zeta function.

The organization of this paper is the following. In Section II we discuss briefly the Riemann zeta function. In Section III, using the fact that the sequence of primes numbers is not zeta regularizable, we conclude that there cannot be a quantum field theory which is able to describe systems with such a spectrum. We also prove that sequences whose asymptotic distributions are not ``too far away" from the asymptotic distribution of the prime numbers are also not zeta regularized (the precise meaning of being``not too far away" will be given in due course). In Section IV, using the fact that the sequence of nontrivial zeros of the Riemann zeta
function is zeta regularizable we propose that quantum field theory can be used to describe a hypothetical system associated with the Hilbert-P\'olya conjecture. In section V we discuss possible relations between the asymptotic behavior of a sequence and the analytic domain of the associated zeta function. Conclusions are given in Section VI. In the paper we use $k_{B}=c=\hbar=1$.

\section{The Riemann zeta function and Hilbert-P\'olya conjecture.}

Prime numbers occur in a very irregular way within
the sequences of all integers in local scales (local scales means on intervals comparable to their mean spacing).
On the other hand, on large scales they are very regular.
The best result that we have concerning the global distribution is the prime number theorem: if $\pi(x)$ is the
number of primes less than or equal to $x$, then $x^{-1}\pi(x)\ln(x)\rightarrow 1$ as $x\rightarrow\infty$~\cite{hadamard,c1,c2}.

Riemann showed how the distribution of the prime numbers is determined by the nontrivial complex zeros of the zeta function. A explicit formula with sums involving the prime numbers and other sums
involving the zeros of the zeta function was presented \cite{ingham,titchmarsh,pat,france}.
We start discussing how the product of all primes appears in a representation of the
Riemann zeta function $\zeta(s)$.

Let $s$ be a complex variable i.e. $s=\sigma+i\tau$ with $\sigma,\tau \in \mathbb{R}$. For $\Re(s)>1$ the Dirichlet series
\begin{equation}
\sum_{n=1}^{\infty}\,\frac{1}{n^{s}}
\label{p2}
\end{equation}
converges absolutely, and uniformly for $\Re(s)\,\geq\,1+\delta$, for all $\delta\,>\,0$. It is possible to show that
\begin{equation}
\zeta(s)=\prod_{p}\,\Biggl(\frac{1}{1-p^{-s}}\Biggr),
\label{p4}
\end{equation}
for $p\,\in \cal{P}$ where $\cal{P}$ is the set of all prime numbers. The product giving by Eq.~(\ref{p4}) is called the Euler product. This is an analytic form of the fundamental theorem of arithmetic, since primes are the multiplicative building block for the natural numbers. Eqs.~(\ref{p2}) and~(\ref{p4}) connect the additive structure in order to generate successive positive integers to this
multiplicative structure. From the convergence of Eq. (\ref{p4}) we obtain that $\zeta(s)$ has no zeros for $\Re(s)>1$.

The Riemann zeta function $\zeta(s)$ is the analytic continuation of the Dirichlet series defined by Eq.~(\ref{p2}) to the whole complex plane. Its unique singularity is the point $s=1$ at which it has a simple pole with residue $1$. Moreover it satisfies the functional equation
\begin{equation}
\pi^{-\frac{s}{2}}\Gamma\biggl(\frac{s}{2}\biggr)\zeta(s)=\pi^{-\frac{(1-s)}{2}}\Gamma\biggl(\frac{1-s}{2}\biggr)\zeta(1-s),
\label{f}
\end{equation}
for $s \in \mathbb{C}\setminus\left\{0,1\right\}$. Let us define the entire function $\xi(s)$ as
\begin{equation}
\xi(s)=\frac{1}{2}s(s-1)\pi^{-\frac{s}{2}}\Gamma\biggl(\frac{s}{2}\biggr)\zeta(s).
\end{equation}
Using the function $\xi(s)$, the functional equation given by Eq.~(\ref{f}) takes the form $\xi(s)=\xi(1-s)$. If $\rho$ is a zero of $\xi(s)$, then by the functional equation so is $1-\rho$. Since $\bar{\xi}(\rho)=\xi(\bar{\rho})$ we have that $\bar{\rho}$ and $1-\bar{\rho}$ are also zeros. The zeros are symmetrically arranged about the real axis and also about the critical line.
Let us write the complex zeros of the zeta function as $\rho=\frac{1}{2}+i\gamma$, $\gamma \in  \mathbb{C}$. The Riemann hypothesis is the statement that all $\gamma$ are real. A weak form of the Riemann hypothesis, namely that there are no zeros of zeta function on the line $\left\{s:\Re(s)=1\right\}$, implies the prime number theorem. We have $\zeta(1+i\tau)\neq 0$ for all $\tau \in \mathbb{R}$.

\section{Asymptotic distributions and zeta regularization}

Some authors~\cite{mussardo,rosu,zyl,pre} formulated the following question: is there a quantum mechanical potential related to the prime numbers? Quantum field theory is the formalism where the probabilistic interpretation of quantum mechanics and the special theory of relativity were gathered to take into account a plethora of phenomena not described by classical physics and quantum mechanics of systems with finite degree of freedom. Therefore it is natural to inquire whether a quantum field theory model can be used to describe some hypothetical physical system with infinite degrees of freedom, whose spectrum is given by the sequence of prime numbers. In the following we will present results that put strong restrictions on the existence of quantum field theory models with a pure prime numbers spectrum. Let us examine this problem in more detail.

Let us discuss the analytic continuation of the prime zeta function and how close of $\pi(x)$ an asymptotic distribution shall be
so that the associated spectral zeta function has the same singular structure as the prime zeta function $P(s)$.

Let $\left\{a_{n}\right\}_{n \in \mathbb{N}}$ be a sequence of nonzero complex numbers. The zeta regularization product is defined as
\begin{equation}
\prod_{n \in \mathbb{N}}\,a_{n}=:\exp\biggl(-\frac{d}{ds}\zeta_{a}(s)|_{s=0}\biggr),
\end{equation}
provided that the spectral zeta function $\sum_{n \in \mathbb{N}}a_{n}^{-s}$ has an analytic extension and is holomorphic at $s=0$. Note that the definition of a zeta regularized product associated with a sequence of complex numbers depends on the choice of $\arg a_{n}$. Therefore some care must be taken to deal with this problem. However, in this paper we assume a sequence of nonzero real numbers, so such a problem is absent.

The prime zeta function $P(s)$, $s=\sigma+i\tau $, for $\sigma,\tau \in \mathbb{R}$, is defined as
\begin{equation}
P(s)=\sum_{\left\{p\right\}}\,p^{-s},  \,\,\Re(s)>1,
\label{21}
\end{equation}
where the summation is performed over all prime numbers~\cite{lan,carl}. We are using again the notation $p\,\in \cal{P}$ where $\cal{P}$ is the set of all prime numbers. The series defined by Eq.~(\ref{21}) converges absolutely when $\sigma>1$.

It is clear that if $P(s)$ is the spectral zeta function of some operator, for our purposes
we have to study the analytic extension of the prime zeta function. Using the Euler formula given by Eq.~(\ref{p4}) we have
\begin{equation}
\ln \zeta(s)=
%-\sum_{\left\{p\right\}}\ln(1-p^{-s})=
\sum_{\left\{p\right\}}\sum_{r=1}^{\infty}\frac{1}{r}\,p^{-rs}, \,\,\,\Re(s)>1.
\label{22}
\end{equation}
Using the definition of the prime zeta function we have
\begin{equation}
\ln \zeta(s)=\sum_{r=1}^{\infty}\frac{1}{r}\,P(rs),  \,\,\,\Re(s)>1.
\label{23}
\end{equation}
Introducing the M\"obius function $\mu(n)$ ~\cite{hardy},
%Let us introduce the M\"obius function $\mu(n)$ defined by~\cite{hardy}
%
%\begin{displaymath}
%\mu(n)=\left\{
%\begin{array}{ll}
%1,\,\,& \mbox{if n is a square-free positive integer}\\
%&\mbox{with an even number of prime factors}\\
%-1,\,\,& \mbox{if n is a square-free positive integer}\\
%&\mbox{with an odd number of prime factors}\\
%0\,\,& \mbox{if n is not square-free}\\
%\end{array}\right.
%\end{displaymath}\\
%
%and $\mu(1)=1$. Using the  M\"obius function $\mu(n)$ 
it is possible to show that the prime zeta function $P(s)$ can be expressed as
\begin{equation}
P(s)=\sum_{k=1}^{\infty}\frac{\mu(k)}{k}\,\ln\zeta(ks), \,\,\,\Re(s)>1.
\label{24}
\end{equation}
The analytic continuation of the prime zeta function can be obtained only using the Riemann analytic
continuation of the Riemann zeta function.
Defining the function $\phi(x)$ as
\begin{equation}
\phi(x)= \sum_{n=1}^{\infty}\,e^{-n^{2}\,\pi\,x},
\label{z2}
\end{equation}
we can write
\beq
\Gamma\biggl(\frac{s}{2}\biggr)\,\,\pi^{-\frac{s}{2}}\,\zeta(s)=\frac{1}{s(s-1)}+\int_{1}^{\infty}dx\,\phi(x)\biggl
(x^{\frac{s}{2}-1}+x^{-\frac{1}{2}(s+1)}\Biggr).
\label{z7}
\eeq
The integral that appears in the Eq. (\ref{z7}) is convergent for
all values of $s$ and therefore the Eq. (\ref{z7}) gives the
analytic continuation of the Riemann zeta function to the whole
complex $s$-plane. The only singularity is the pole at $s=1$.
Consequently, for the prime zeta function $P(s)$, $s=1/k$
is a singular point for all square free positive integers $k$. This sequence limits to $s=0$. All
points on the line $\Re(s)=0$ are limit points of the singularities of $P(s)$, so that the line $\Re(s)=0$ is a
natural boundary of $P(s)$. The prime zeta function can be analytically extended only in the strip
$0<\sigma\leq 1$. This result, that the prime zeta function cannot be
continued beyond the line $\sigma=0$ was obtained by Landau and Walfisz~\cite{lan} and discussed by Fr\"oberg~\cite{carl}. 
Since we need to compute the derivative of the spectral prime zeta function at $s=0$, the assumption of a prime number spectrum implies in an ill-defined functional integral.

It is natural to enquire which other sequences, besides the prime numbers, cannot be zeta regularized and therefore cannot be associated with the spectrum of self-adjoint operators in quantum field theory. In a recent paper Andrade~\cite{andrade} proved that there are other number-theoretical sequences of numbers that cannot be the spectrum of an operator describing a bosonic system with infinite degrees of freedom described by quantum field theory. Here we extend this result to sequences whose asymptotic distributions are not ``too far away" from the asymptotic distribution of the prime numbers.

\begin{theorem}
Suppose that $Q=\{q_{1}\leq q_{2}\leq \cdots\}\subset (0,+\infty)$, $q_{n}\rightarrow\infty$ as $n\rightarrow\infty$, and define
\begin{equation}
\pi_{Q}(x)=\#\{n\,|\,q_{n}\leq x\}.
\end{equation}
If $|\pi(x)-\pi_{Q}(x)|=O((\ln x)^{k})$ for $x$ large enough and $k>0$, then $\zeta_{Q}(s)=\sum_{n}\frac{1}{q_{n}^{s}}$ has the same singularities as the prime zeta function and therefore has no analytic extension at $s=0$.
\end{theorem}

\begin{proof}
In view of the assumptions
\begin{equation}
\zeta_{Q}(s)=\sum_{n}\frac{1}{q_{n}^{s}}, \,\,\,\,\text{for\,}Re(s)>1,
\end{equation}
and this function is analytic for $Re(s)>1$. In this region we can use the following representations for $\zeta_{Q}(s)$ and the prime zeta function $P(s)$, respectively,
\begin{equation}\label{aa}
\zeta_{Q}(s)=\int_{a}^{\infty}\frac{1}{t^{s}}d\big(\pi_{Q}(t)\big)
\end{equation}
and
\begin{equation}\label{aaa}
P(s)=\int_{a}^{\infty}\frac{1}{t^{s}}d\big(\pi(t)\big),
\end{equation}
where $0<a<\min\{2,q_{1}\}$ and the integrals are Riemann-Stieltjes integrals. Define for $Re(s)>1$,  $G(s)$ as
\begin{equation}
G(s)=P(s)-\zeta_{Q}(s).
\end{equation}
This function is analytic for $Re(s)>1$. Moreover, in view of Eqs.~(\ref{aa}) and~(\ref{aaa}) in this region we can write
\begin{equation}
G(s)=\lim_{R\to\infty}\int_{a}^{R}\frac{1}{t^{s}}d\big(\pi(t)-\pi_{Q}(t)\big).
\end{equation}
Since $\pi(s)$ and $\pi_{Q}(s)$ have bounded variations in finite intervals, integrating by parts we have
\begin{equation}
G(s)=\lim_{R\to\infty}\left[\frac{\pi(t)-\pi_{Q}(t)}{t^{s}}\right]_{a}^{R}+s\int_{a}^{R}\frac{\pi(t)-\pi_{Q}(t)}{t^{s+1}}dt.
\end{equation}
The first limit vanishes and we conclude that
\begin{equation}
G(s)=s\int_{a}^{\infty}\frac{\pi(t)-\pi_{Q}(t)}{t^{s+1}}dt,
\end{equation}
for $Re(s)>1$.
In view of the assumptions, the above expression for $G(s)$ is well defined for $Re(s)>0$. Using Lebesgue's dominated convergence theorem it is trivial to verify that $G(s)$ is differentiable for $Re(s)>0$. All together we have $\zeta_{Q}(s)=P(s)-G(s)$ for $Re(s)>0$,  with $G(s)$ analytic in this region. This concludes the proof.

\end{proof}

The main point of Riemann's paper is that the two sequences, of prime numbers on one hand and the
Riemann zeros on the other hand are connected. The erratic behavior of the primes is encoded in the nontrivial zeros
of the zeta function. Although we also expect this erratic behavior for the zeros, these two sequence of numbers have totally
distinct behavior with respect to being the spectrum of a linear differential operator of a system with infinite degrees of freedom.

In the next Section, we discuss if the sequence of nontrivial zeros of the Riemann zeta function can be associated with the spectrum of a 
linear differential operator of a system with countably infinite number of degrees of freedom described by quantum field theory.

\section{The Super-Zeta or Secondary Zeta Function and its Analytic Extension}

Let $(M,g)$ be a compact, Riemaniann $C^{\infty}$ manifold, with metric $g=(g_{ij})$ on $M$. We assume that $M$ is connected and dim$\,M=d$. Let $D$ be a generic elliptic operator acting on a neutral scalar field $\varphi$, both defined in $M$. We assume that the behavior of the fields at infinity is such that the compactification is possible. In order to obtain the correlation functions of the theory one should construct the generating functional $Z[h]$. Defining an appropriate kernel $K(m_{0};\,x-y)$, the generating functional $Z[h]$ is formally defined by the following functional integral:
\begin{equation}
Z[h]=\int [d\varphi]\,\, \exp\left(-S_{0} - \int d^{d}x\sqrt{g(x)}\,
h(x)\varphi(x)\right),
\label{1}
\end{equation}
where $g=\det(g_{ij})$ and the action that usually describes a scalar field is
\begin{equation}
S_{0}(\varphi)=\int d^{d}x\,d^{d}y\,\sqrt{g(x)}\sqrt{g(y)}\varphi(x)K(m_{0};\,x-y)\varphi(y).
\label{2}
\end{equation}
In Eq.~(\ref{1}), $[d\varphi]$ is a appropriate measure, formally given by $[d\varphi]=\prod_{x} d\varphi(x)$. The term $m_{0}^{2}$ is the (bare) 
mass squared of the model. Finally, $h(x)$ is a smooth function introduced to generate the Schwinger functions of the theory. One can define the functional $W[h]=\ln Z[h]$ which
generates the connected Schwinger functions.
In order to obtain a well-defined object, we need to regularize a determinant associated with the operator $D$, since $W[0]= - 1/2 \ln\det D$. 
A similar situation arises when one calculates the one-loop effective action for non-Gaussian functional integrals.

In the following we use the Minakshisundaram-Pleijel zeta function~\cite{mp}.
The standard technique of the spectral theory of elliptic operators implies the existence of a complete
orthonormal set $\left\{f_{k}\right\}_{k=1}^{\infty}$ such that the eigenvalues satisfy: $0\leq \lambda_{1}\leq\lambda_{2}\leq
\,...\,\leq\lambda_{k}\rightarrow\infty$, when $k\rightarrow\infty$ where the zero eigenvalue must be omitted
(eigenvalues being counted with their multiplicities).
In the basis $\left\{f_{k}\right\}$ the operator $D$ is represented by an infinite diagonal matrix $D=\diag\,(\lambda_{1}, \lambda_{2},...)$.
Therefore the generic operator $D$ satisfies $D f_{n}(x)=\lambda_{n}f_{n}(x)$. The spectral zeta function associated with the operator $D$ is defined as
\begin{equation}
\zeta_{D}(s)=\sum_{n}\frac{1}{\lambda_{n}^{s}},\,\,\,\,\,\,\,\Re(s)>s_{0},
\label{imp}
\end{equation}
for some $s_0$. Formally we have
\begin{equation}
-\frac{d}{ds}\zeta_{D}(s)|_{s=0}=\ln\det D.
\label{imp2}
\end{equation}

In order to regularize the determinant or the functional integral it is necessary to perform an analytic continuation
of the spectral zeta function from some half-plane, i.e., for sufficient large positive $\Re(s)$ into the whole
complex plane. The spectral zeta function must be analytic in a complex neighborhood of the origin, i.e., $s=0$.
This method can also be used non-Gaussian functional integrals when one calculates the one-loop effective action. In addition,
we need to use scaling properties, i.e.,
\begin{equation}
\frac{d}{ds}\zeta_{\mu^{2}D}(s)|_{s=0}=\ln\,\mu^{2}\zeta_{D}|_{s=0}+\frac{d}{ds}\zeta_{D}|_{s=0}.
\label{17}
\end{equation}

Let us discuss the construction of the so-called super-zeta or secondary zeta function built over the Riemann zeros, i.e., 
the nontrivial zeros of the Riemann zeta function~\cite{gui,delsarte,ivic,voros1}. In the following, for completeness, we present the analytic extension for $G_{\gamma}(s)$, defined by
\begin{equation}
G_{\gamma}(s)=\sum_{\gamma > 0}\gamma^{-s}, \,\,\,\Re(s)>1,
\label{zetazero}
\end{equation}
assuming the Riemann hypothesis. We are following the Ref.~\cite{delsarte}. Using the definition given by Eq.~(\ref{zetazero}) we get
\begin{equation}
\Gamma\biggl(\frac{s}{2}\biggr)\pi^{-\frac{s}{2}}G_{\gamma}(s)=\int_{0}^{\infty}\,dx\,
x^{\frac{s}{2}-1}\sum_{\gamma>0}\,e^{-\pi\gamma^{2}x}.
\label{zetazero2}
\end{equation}
Let us split the integral that appears in Eq.~(\ref{zetazero2}) in the intervals $[0,1]$ and $[1,\infty)$, and define the functions
\begin{equation}
A(s)=\int_{0}^{1}\,dx\, x^{\frac{s}{2}-1}\sum_{\gamma>0}\,e^{-\pi\gamma^{2}x}
\label{zetazero3}
\end{equation}
and
\begin{equation}
B(s)=\int_{1}^{\infty}\,dx\, x^{\frac{s}{2}-1}\sum_{\gamma>0}\,e^{-\pi\gamma^{2}x}.
\label{zetazero4}
\end{equation}
Note that $B(s)$ is an entire function. To proceed let us use that
\beq
\sum_{\gamma>0}\,e^{-\pi\gamma^{2}x} = -\frac{1}{2\pi\sqrt{x}}\sum_{n=2}^{\infty}
\frac{\Lambda(n)}{\sqrt{n}}\,e^{-\frac{(\ln n)^{2}}{4\pi x}}+\,e^{\frac{\pi x}{4}}-\frac{1}{2\pi}\int_{0}^{\infty}dt\,e^{-\pi x t^{2}}\Psi(t),
\label{zetazero55}
\eeq
where the function $\Psi(t)$ is given by
\begin{equation}
\Psi(t)=\frac{\zeta'(\frac{1}{2}+i t)}{\zeta(\frac{1}{2}+i t)}+\frac{\zeta'(\frac{1}{2}-i t)}{\zeta(\frac{1}{2}-i t)}.
\label{zetazero5}
\end{equation}
Substituting Eq.~(\ref{zetazero55}) in~(\ref{zetazero3}) we get that $A$-function can be written as
\begin{equation}
A(t)=A_{1}(t)+A_{2}(t)+A_{3}(t),
\label{zetazero6}
\end{equation}
where
\begin{equation}
A_{1}(s)=-\frac{1}{2\pi}\int_{0}^{1}\,dx\,x^{\frac{s}{2}-\frac{3}{2}}\biggl(\sum_{n=2}^{\infty}
\frac{\Lambda(n)}{\sqrt{n}}\,e^{-\frac{(\ln n)^{2}}{4\pi x}}\biggr),
\label{zetazero7}
\end{equation}
\begin{equation}
A_{2}(s)=\int_{0}^{1}\,dx\,x^{\frac{s}{2}-1}\,e^{\frac{\pi x}{4}}
\label{zetazero8}
\end{equation}
and finally
\begin{equation}
A_{3}(s)=-\frac{1}{2\pi}\int_{0}^{1}\,dx\,x^{\frac{s}{2}-1}\,\biggl(\int_{0}^{\infty}\,e^{-\pi x t^{2}}\Psi(t)\biggr).
\label{zetazero9}
\end{equation}
Changing variables in the $A_{1}(s)$, i.e., $x\rightarrow1/x$ we get
\begin{equation}
A_{1}(s)=-\frac{1}{2\pi}\int_{1}^{\infty}\,dx\,x^{-\frac{s}{2}-\frac{1}{2}}
\biggl(\sum_{n=2}^{\infty}\frac{\Lambda(n)}{\sqrt{n}}\,e^{-\frac{x(\ln n)^{2}}{4\pi}}\biggr).
\label{zetazero10}
\end{equation}
It is clear that $A_{1}(s)$ is an entire function of $s$. Let us define $\Phi(s)$ as
\begin{equation}
\Phi(s)=A_{1}(s)+B(s).
\label{zetazero11}
\end{equation}
Using Eqs.~(\ref{zetazero4}),~(\ref{zetazero6}),~(\ref{zetazero8}),~(\ref{zetazero9}) and~(\ref{zetazero11}) we can write expression~(\ref{zetazero2}) as
\begin{equation}
\Gamma\biggl(\frac{s}{2}\biggr)\pi^{-\frac{s}{2}}G_{\gamma}(s)=\Phi(s)+A_{2}(t)+A_{3}(t).
\label{zetazero12}
\end{equation}
Since $\Phi(s)$ is an entire function and we have the integrals that define $A_{2}(t)$ and $A_{3}(t)$, the above formula is
the analytic extension of the secondary zeta function. The function
$G_{\gamma}(s)$ is a meromorphic function of $s$ in the whole complex plane with double pole at
$s=1$ and simple poles at $s=-1,-2,..,-(2n+1),..\,$. Therefore $(s-1)^{2}G_{\gamma}(s)(\Gamma(s))^{-1}$ is an entire function.
From the above discussion we have that the spectral determinant associated with the zeta zeros is well defined, since it is
possible to find an analytic continuation of $G_{\gamma}(s)$ to a meromorphic function in the whole
complex $s$-plane and also analytic in a complex neighborhood of the origin, i.e., $s=0$. A more detailed development of this interesting viewpoint can be found in~\cite{superzeta}.

We conclude that there is a large class of hypothetical systems with countably infinite number of degrees of freedom, described by
self-adjoint operators whose spectra can be given by the sequence of the nontrivial zeros of the Riemann zeta function.

The question that confront us is the following. Is it possible to extend the last result to sequences whose asymptotic distributions
are not ``far way" from the asymptotic distribution of the nontrivial zeta zeros? Now we discuss possible relations between the asymptotic behavior of a numerical sequence and the analytic domain of the associated zeta function. This is the topic of the next Section.

\section{The approximations for the super-zeta function and the prime zeta function using asymptotic distributions}

The aim of this section is to propose an approximation for a zeta function of a numerical sequence by means of its
asymptotic distribution. For the case of the nontrivial zeros of Riemann zeta function, the analytic domain and polar
structure of the approximated zeta function is the same as the super-zeta function for $\Re(s)> -\frac{3}{2}$.
However, for the sequence of prime numbers, this approximation fails to reproduce the analytic structure of the prime zeta function.

Let $\lambda_1<\lambda_2<\dots$ be the spectrum of a linear operator $D$.
We discussed before a necessary condition for a quantum field theory model to reasonably describe a physical system possessing such a spectrum: the spectral zeta-function associated with such a sequence shall have an analytic continuation which includes the origin.
What properties of the sequence $(\lambda_n)$ determine the existence of this analytic continuation?

The spectral zeta function associated with the sequence $\lambda_n$ can be expressed as
\beq
\zeta_{D}(s) = \sum_{n=1}^{\infty}\,\lambda_n^{-s} = \lim_{m\to\infty}\sum_{n=1}^{m}\,\lambda_n^{-s}.
\label{zeta}
\eeq
In order to define the above zeta function by means of an integral, let us introduce the following counting function
$$
F(t) = \#\{\lambda_n\,|\,\lambda_n < t\}
$$
that is, $F(t)$ is the number of elements in the sequence $\lambda_n$ which are less than $t$.
The direct use of the definition of the Riemann-Stieltjes integral yields
\beq
\sum_{n=1}^{m}\,\lambda_n^{-s} = \sum_{n=1}^{k-1}\,\lambda_n^{-s} + \int_{a}^{b}\,t^{-s}\,dF(t),\,\,\,\lambda_{k-1}\leq a < \lambda_k\,,\lambda_{m}\leq b < \lambda_{m+1}.
\eeq
Therefore, the spectral zeta function $\zeta_D(s)$ can be expressed as
\beq
\zeta_{D}(s) = \sum_{n=1}^{k-1}\,\lambda_n^{-s} + \int_{\lambda_k}^{\infty}\,t^{-s}\,dF(t).
\label{zeta:RS}
\eeq
Such a formula is valid in the region of convergence of the series given by Eq.~(\ref{zeta}). Since the finite sum in the
right-hand side of Eq.~(\ref{zeta:RS}) is analytic over the whole complex $s$-plane, the qualitative behavior of
the analytic extension of the above function is determined by the above Riemann-Stieltjes integral.

Since it is a hard task to compute $F(t)$ in general, one may resort to asymptotic expansions, when they are
available. For the case where the sequence $\lambda_n$ is given by the sequence of nontrivial zeros of the Riemann zeta function, we have
\beq
F(t) \approx t\ln t\,,t\to\infty.
\eeq
Using this approximation to $F(t)$ we get, for the spectral zeta function of the nontrivial Riemann zeros
\beq
\tilde{\zeta}_{D}(s) = \sum_{n=1}^{k-1}\,\lambda_n^{-s} + \int_{\lambda_k}^{\infty}\,t^{-s}\,d(t \ln t).
\eeq
Since the integrator $t\ln t$ is smooth over the interval $(0,\infty)$, the Riemann-Stieltjes integral
coincides with the usual Riemann integral and we have
$$
\int_{\lambda_k}^{\infty}\,t^{-s}\,d(t \ln t) =\int_{\lambda_k}^{\infty}\,t^{-s}\,\frac{d}{dt}(t\ln t)\,dt.
$$
Thus, one gets for $\Re(s) > 1$
\beq
\tilde{\zeta}_{D}(s) = \sum_{n=1}^{k-1}\,\lambda_n^{-s} + \frac{\lambda_{k}^{-(s-1)}}{(s-1)^2} +
\frac{\lambda_{k}^{-(s-1)}}{(s-1)}\left(1 + \ln\lambda_{k}\right).
\eeq
Analytic continuation of the above result gives the same pole structure of the super-zeta $G_{\gamma}(s)$ in the neighborhood of $\Re(s) \geq 0$.

Now let us apply the same reasoning when the sequence $\lambda_n$ is given by the sequence of primes. In this case,
the integrator $F(t)$ is just the prime counting function $\pi(t)$. Asymptotically we have
\beq
F(t) = \pi(t) \approx \frac{t}{\ln t}\,,t\to\infty.
\eeq
Using this approximation, the spectral zeta function for the sequence of prime numbers is given by
\beq
\hat{\zeta}_{D}(s) = \sum_{n=1}^{k-1}\,\lambda_n^{-s} + \int_{\lambda_k}^{\infty}\,t^{-s}\,d\left(\frac{t}{\ln t}\right).
\eeq
The integrator is smooth over the interval $(\lambda_k,\infty)$, $\lambda_k > 1$, then the again Riemann-Stieltjes
integral coincides with the usual Riemann integral. The final result is
\beq
\hat{\zeta}_{D}(s) = \sum_{n=1}^{k-1}\,\lambda_n^{-s} + \E_{1}[(s-1)\ln\lambda_k]
- \frac{1}{\ln\lambda_k}\E_{2}[(s-1)\ln\lambda_k],
\eeq
where $\E_{m}(z) $ is the usual exponential integral. Using the recurrence relation
\beq
E_{n+1}(z)=\frac{1}{n}[e^{-z}-z\,E_{n}(z)]\,\,\, n=1,2,3..
\eeq
we have
\begin{equation}
\hat{\zeta}_{D}(s)= \sum_{n=1}^{k-1}\,\lambda_n^{-s}+s\,E_{1}[(s-1)\ln \lambda_{k}]-\frac{1}{\lambda_{k}^{s-1}\,\ln\,\lambda_{k}}.
\end{equation}

Using the analytic continuation of such a function in the above result yields multi-valued functions with a branch point at $s=1$. Comparison with the prime zeta function $P(s)$ defined by Eq.~(\ref{21}) reveals that $\hat{\zeta}_{D}(s)$ fails to reproduce qualitatively the same singular structure of $P(s)$ in the neighborhood of $\Re(s) \geq 0$.

\section{Conclusions}

The analytic function which encodes information on the prime factorization of integers and distribution of primes is the Riemann zeta function $\zeta(s)$ because in the region of the complex plane where it converges absolutely and uniformly, the
product of all prime numbers appears in a representation of the Riemann zeta function $\zeta(s)$. The Riemann hypothesis claims that all nontrivial zeros of the Riemann zeta function $\zeta(s)$ lie on the critical line $\Re(s)=1/2$. This hypothesis
makes a deep connection between primes numbers and zeros of analytic functions. Hilbert and P\'olya suggested that there might be a spectral interpretation of the the non-trivial zeros of the Riemann zeta function. The corresponding operator must be self-adjoint.

Non-separability in wave mechanics leads us to chaotic systems in the limit of short wavelengths. Therefore one cannot disregard the possibility that a quantum field model described by a nonseparable wave equation for a given boundary condition is able to reproduce the statistical properties of the nontrivial zeta zeros. Other possibility is to consider a field theory in disordered medium, such as wave equations in amorphous solids. We conjecture that a system with countably infinite number of degrees of freedom with randomness can be used to achieve further progress in the Hilbert-P\'olya conjecture.

In this paper we proved that for systems described by a self-adjoint operator where the spectrum is given by the prime numbers, the associated functional integral cannot be constructed. The impossibility of extending the definition of the analytic function $P(s)$ to the half-plane $\sigma<0$ means that the functional determinant cannot be defined, in virtue that the prime numbers sequence is not zeta regularizable. Second, prompted by this result, we prove that sequences whose asymptotic distribution is close to that of the prime numbers are also not zeta regularized. Therefore, these sequences cannot be
associated with the spectrum of a self-adjoint operator in quantum field theory. Next, we obtained that in principle hypothetical physical systems described by self-adjoint operators with a spectrum given by the sequence of nontrivial zeros of the Riemann zeta function can be described by a quantum field theory model in the functional integral approach. We discuss possible relations between the asymptotic behavior of a sequence and the analytic domain of the associated zeta function.

Finally, more can be said. Suppose that the generic operator $D$ is the usual Laplacian defined in a Euclidean manifold.
If one compares the Weyl law for the asymptotic series for the density of eigenvalues
of the Laplacian operator in a three-space, $N(\omega)= V\omega^{2}/2\pi^{2}+..$, where $V$ is the volume of the three-space~\cite{cou,can} with the asymptotic distribution of the zeta zeros or the asymptotic distribution of prime numbers that
follows from the prime number theorem, we get quite different regimes.
For usual quantized systems we get that $\omega_{n}$ must increase at a rate dictated by the Weyl law.
For the prime numbers, using the prime number theorem, we get the asymptotic regimes $p_{n}\sim\,n\ln n$.
If the zeros $\rho=\beta+i\gamma$, with $\gamma>0$ are arranged in a sequence $\rho_{k}=\beta_{k}+i\gamma_{k}$ so that
$\gamma_{k+1}>\gamma_{k}$, then $|\rho_{k}|\sim\gamma_{k}\sim 2\pi k/\ln k$ as $k\rightarrow \infty$.
Therefore for the zeros of the zeta function we get $\gamma_{k}\sim\,k/\ln k$. Although both asymptotic
regimes are quite different from the usual systems, it is certainly a step forward to explain why the zeta regularization works for the sequence of zeta zeros and fails for the sequence of prime numbers and to sequences whose asymptotic distributions are not ``too far away" from the asymptotic distribution of the prime numbers.

\section{Acknowlegements}

This paper was supported by Conselho Nacional de Desenvolvimento
Cientifico e Tecnol{\'o}gico do Brazil (CNPq).

\end{document}